\documentclass[11pt,a4paper]{article}%
\pdfoutput=1
\usepackage[UKenglish]{babel}
\usepackage{lmodern}
\usepackage[T1]{fontenc}
\usepackage[protrusion=true,expansion=true]{microtype}
\usepackage{geometry}
\usepackage{graphicx}
\usepackage{amsthm}
\usepackage[cmex10]{amsmath}
\usepackage{natbib}
\usepackage{amssymb}%
\usepackage{dcolumn}%
\usepackage{booktabs}%
\usepackage{bm}%
\usepackage{mathtools}%
\usepackage{enumitem}
\usepackage[format=hang, labelfont=bf,
tableposition=top, figureposition=bottom]{caption}
\usepackage{parskip}
\usepackage{setspace}
\usepackage{fullpage}
\usepackage{float}
\makeatletter
\def\thm@space@setup{%
  \thm@preskip=\parskip \thm@postskip=0pt
}
\makeatother
\newcolumntype{d}[1]{D{.}{\cdot}{#1}}
\numberwithin{equation}{section}
\mathtoolsset{showonlyrefs}
\theoremstyle{plain}
\newtheorem{theorem}{Theorem}[section]

\newtheorem{proposition}[theorem]{Proposition}

\title{\textbf{
A Short Solution to the Many-Player Silent Duel with Arbitrary Consolation Prize
}}

\author{Steve Alpern%
\thanks{Steve Alpern is at the Warwick Business School,
University of Warwick, Coventry CV4 7AL, UK %
\texttt{ Steve.Alpern@wbs.ac.uk}}%
\and%
J. V. Howard%
\thanks{J. V. Howard is at the London School of Economics,
Houghton Street, London WC2A 2AE, UK %
\texttt{  j.v.howard@lse.ac.uk}}%
}%
\date{1 December, 2017}
\begin{document}
\maketitle%
\thispagestyle{empty}%
\pagestyle{plain}
\begin{abstract}
The classical constant-sum `silent duel' game had two antagonistic marksmen walking towards each other. A more friendly formulation has two equally skilled marksmen approaching targets at which they may silently fire at distances of their own choice. The winner, who gets a unit prize, is the marksman who hits his target at the greatest distance; if both miss, they share the prize (each gets a `consolation prize' of one half). In another formulation, if they both miss they each get zero. More generally we can consider more than two marksmen and an arbitrary consolation prize. This non-constant sum game may be interpreted as a research tournament where the entrant who successfully solves the hardest problem wins the prize. We give the first complete solution to the many-player problem with arbitrary consolation prize: moreover (by taking particular values for the consolation prize), our theorem incorporates various special results in the literature, and our proof is simpler than any of these.
\end{abstract}
\newpage
\section{Introduction}
Suppose that a number of equally skilled competitors attempt a task whose level of difficulty can be varied -- for example lifting a heavy barbell (vary the weight), jumping over a bar (vary the height), or proving a mathematical theorem (vary the theorem). Each competitor chooses his own level of difficulty and is allowed only one attempt, not knowing what levels the others are attempting nor whether they succeed or fail. The player who is successful at the highest level wins, or shares the prize if there is a tie. Of course, this protocol does not follow normal weight-lifting or high-jumping competition rules. However, it does give a simple model of a research tournament where entrants compete to find the best solution to a problem posed by a firm -- better solutions being more difficult to find.

Another version of the problem would have a line of marksmen walking towards a set of targets (one for each player). They can choose to fire at any distance, and the attempt will either be a hit or a miss. The shots are silent, so the other players do not know who has fired, nor whether anyone has hit the target. A marksman who hits his target from the greatest distance is the winner, and again the prize is shared if there is a tie. If all fail (everyone misses), we could again say the prize is shared, which keeps it a constant-sum game, or give no prize, or give everyone a small consolation prize. (If all the players start by contributing to the prize pool, sharing the prize when all fail is reasonable.) Giving no consolation prize is the standard assumption for prize competitions (research tournaments).

When there are just two players we can assume the target is the other player, and they walk steadily towards each other. The game then becomes the classic game of timing known as the `silent duel'. The purpose of this note is to give a short and simple solution to this problem in the many-person case and for any level of the consolation prize.
\par
Games of timing were first extensively studied at the RAND Corporation during the period 1948-1952, though some of the memorandums were only declassified later. \cite{Radzik1996}, who provides a detailed history (as well as more advanced analysis), cites David Blackwell together with M. Shiffman, M. A. Girshik, L. S. Shapley, R. Bellman, and I Glicksberg as some of the major researchers at that time. The silent duel is the exemplar of games of timing and is a staple example in the major texts on mathematical game theory (e.g. \cite{Owen1995}, \cite{BasarEtal1982}, and \cite{Garnaev2000}). The area has seen significant development (noisy duels, many bullets, limited time horizon). Most of these extensions remain two person constant-sum games, with the exceptions of \cite{Sakaguchi1978}, \cite{HenigONeill1992}, \cite{BastonGarnaev1995}, and \cite{PresmanSonin2006} (who also allow for players with different marksmanship skills).
\section{The model and solution}
We have $n$ identical players, called marksmen. A pure strategy for each is a firing distance, $x$. We assume the probability of missing is $1$ at maximum distance, and declines strictly and  continuously as the range decreases, becoming zero when the range is zero. So we can measure the firing distance by the probability of missing, $x\in [0,1]$. In a mixed strategy the firing distance will be a random variable $X$ having a probability distribution specified by a (cumulative) distribution function $G(x)$.
\par
We now introduce a new random variable, the `score' $Y$. The score is the same as the firing distance, $x$, if the player hits the target (which has probability $1-x$), but equals $-1$ if he misses (which has probability $x$). So the rules imply that the player with the highest score (who hits the target at the greatest distance) wins the prize, which we take to be $1$. If there is a tie at a non-negative score the prize is divided equally among the winners. The event that all the marksmen miss, and get the same score of $-1$, has positive probability, and in this case we say that everyone gets the same `consolation prize' of $c\leqslant 1$. If $c=1/n$, which corresponds to the players splitting the prize, the game has constant sum $1$. More specifically, if $n=2$ and $c=1/2$, the game is the original silent duel. Setting $c$ to zero means that no prize is given for failure. A valid strategy is to fire at point-blank range, hoping all the other players have fired and missed. In this case the score is $0$, but this is still better than a miss, so a miss must be assigned a negative score: the choice of $-1$ is arbitrary -- any negative number would do equally well.
\par
A pure strategy for the firing distance (fire at distance $x$) gives rise to a two-atom distribution for the score $Y$ ($Y = x$ with probability $1 - x$, and $Y = -1$ with probability $x$). A mixed strategy for the firing distance gives a mixture of these two-atom distributions for $Y$. Let $F(y)$ be the distribution function for the score. $F$ will have an atom of probability, $p$, at $-1$, but all the remaining probability will lie within $[0, 1]$. So $F(0) \geqslant p$ and $F(1) = 1$.
\par
We note that if there is a density function $g(x)$ for the firing distance, then (aside from the atom at $x=-1$) there is a density $f(y)$ for the score, satisfying
\begin{equation}
\label{eq:fandg}
f(x) =(1 - x) \, g(x) \text{.}
\end{equation}
\begin{theorem}
\label{thm:sdsoln}
Suppose $n \in \mathbb{N}$ with $n \geqslant 2$, and $c \in \mathbb{R}$ with $0 \leqslant c < 1$. Let $1/p$ be the unique solution in $(1, \infty)$ of the polynomial
\begin{align*}
\left( \frac{1}{p} \right) ^{n} &= 1 - nc + n \left( \frac{1}{p} \right) \text{.}
\end{align*}%
Then the $n$-player silent duel game with compensation $c$ has a unique symmetric equilibrium in which each player has overall probability $p$ of missing, with equilibrium payoff $v = p^{n - 1}$, and with score distribution $F(y)$ supported on $\{ -1 \} \cup [0, b]$, where
\begin{equation}
\label{eq:F(y)}
F(y) = p \left( \sqrt[n-1]{\frac{1 - cy}{1 - y}} \right) \text{, for }%
0 \leqslant y \leqslant b = \frac{1 - v}{1 - cv} \text{.}
\end{equation}
\end{theorem}
\begin{proof}
Suppose that there is a symmetric equilibrium in which every player achieves expected value $v$, has the score distribution $F(y)$, and has probability of missing $p$. We note first that $F$ must be continuous on $[0, 1]$, because if there was an atom of probability at some distance $y$ this would imply that that the firing distribution $G$ also had an atom of probability at $y$. But then if one player deviated by moving this atom to $y + \epsilon$, for sufficiently small $\epsilon$, she would obtain a larger expected payoff. This argument does not apply when $y = 1$, but the strategy of firing at maximum distance (when you are certain to miss) is dominated by firing at any shorter range.
\par
If the first $n-1$ players adopt the score distribution $F(y)$, and the $n$'th player chooses to fire at any distance $y \in [0, 1]$, she should get an expected payoff no greater than $v$. As $F$ does not have an atom at $y$, the chance that she ties with another player is zero. If she hits the target (which has probability $1-y)$ then she wins with payoff $1$ if $y$ is greater than the other players' successful firing distances, that is with probability $F^{n-1}(y)$. If she misses, she still gets the consolation prize $c$ if all the other players also miss. So her expected payoff $v$ is made up of these two terms, and so for $0 \leqslant y < 1$ we have
\begin{equation}
\label{eq:cfirst}
(1 - y) \, F^{n-1}(y) + cyp^{n-1} \leqslant v
\text{,}
\end{equation}
giving
\begin{equation}
F(y) \leqslant \sqrt[n - 1]{\frac{v - cyp^{n-1}}{1 - y}}
\text{.}
\label{eq:csecond}
\end{equation}%
\par
By assumption, there must be some $y$-value(s) which give Player $n$ the equilibrium value $v$. In fact, Player $n$ will be prepared to use $y$ values only where the inequality is binding. If the inequality is not binding at some point $y$, there will be an interval including $y$ in which the inequality is not binding. Player $n$ will not be prepared to place any probability in this interval, and since we are assuming we have a symmetric equilibrium, the other players will do the same. So where $F$ is below the bounding curve it must be horizontal. If it touches the curve again at some larger value of $y$, there must be an atom of probability to jump $F$ up to meet the curve again. But we saw that $F$ must be atomless. So the bound must be tight up to some value $b$ at which $F(b) = 1$. Similarly there cannot be an atom of probability at $0$, so $F(0) = p$.
\par
Taking $y = 0$ in the equation for $F$ gives
\begin{equation}
\label{eq:F(0)}
v = F^{n-1}(0) = p^{n - 1}
\end{equation}%
This establishes that if there is a symmetric solution it must be of the form claimed. It remains to find values for $p$ and $b$.
\par
To find $p$, we note that the expected total amount paid out must equal $nv$, but we also know that it is either $1$ if someone hits the target or else $nc$ if everyone misses, so
\begin{align*}
%\label{eq:value}
nv &= \left(1 - p^{n}\right)1 +p^{n}(nc) \text{,}
\\
np^{n - 1} &= 1 - p^{n} + ncp^{n} \text{,}
\\
\frac{n}{p} &= \left( \frac{1}{p} \right) ^{n} - 1 + nc \text{,}
\\
\left( \frac{1}{p} \right) ^{n} &= 1 - nc + n \left( \frac{1}{p} \right) \text{.}
\end{align*}%
Considered as an equation in $z = 1/p$, we are looking for the intersection of a line of slope $n$ with the curve $z^{n}$ which has slope $n$ at $1$. Clearly there will be a unique solution greater than $1$ on the positive reals provided $c$ is less than $1$.
\par
For $b$, we solve $F(y) = 1$ to get $b = (1 - v)/(1 - cv)$.
\par
$F(y)$ is differentiable on $[0, b]$ to give a score density function $f(y)$. We can then find the firing distance density as $g(x) = f(x)/(1 - x)$. A direct calculation then shows that this density does integrate to $1$ on $[0, b]$, and that if every player uses this density for the firing distance, none will have an incentive to deviate unilaterally.
\end{proof}
\subsection{Special cases}
\begin{proposition}
For the constant-sum case $c=1/n$ (share the prize if all fail), the
equilibrium has the score distribution
\begin{equation}
\label{eq:z=1/k}
F(y) = \sqrt[n - 1]{ \frac{n - y}{n^{2}(1 - y)} } \text{,} \quad
\text{for } 0 \leqslant y \leqslant b = \frac{n}{n + 1} \text{,}
\end{equation}
and the density for the firing distance is
\begin{equation*}
g(x) = \frac{1}{\sqrt[n - 1]{n^{2} (1 - x)^{2n - 1} (n - x)^{n - 2}}} \text{ .}
\end{equation*}
\end{proposition}
\begin{proof}
If $c=1/n$, the game has constant sum $1$, so we have $v=1/n$ by the symmetry of the players. Then $p = 1/\sqrt[n - 1]{n}$, which turns \eqref{eq:F(y)} into \eqref{eq:z=1/k}. $g(x)$ is calculated as $F'(x)/(1 - x)$.
\end{proof}
\par
So
\begin{equation}
g(x) = \frac{1}{4(1 - x)^{3}}
\end{equation}
gives the density of the firing distance $x$ in the classic silent duel ($n = 2$, $c = 1/2$). The solution was published in \cite{karlin1959}. A recent reference is \cite{Owen1995}. \cite{AlpernHoward2017} give an alternative treatment based on Distribution Ranking Games.
%
%\newpage
\par
Figure~\ref{fig:constsum} plots the firing distance densities for the cases $n = 2$ (solid line), $n = 4$ (dashed line), and $n = 6$ (dotted line).
\begin{figure}[H]
\centering
\includegraphics[scale=0.7, clip]{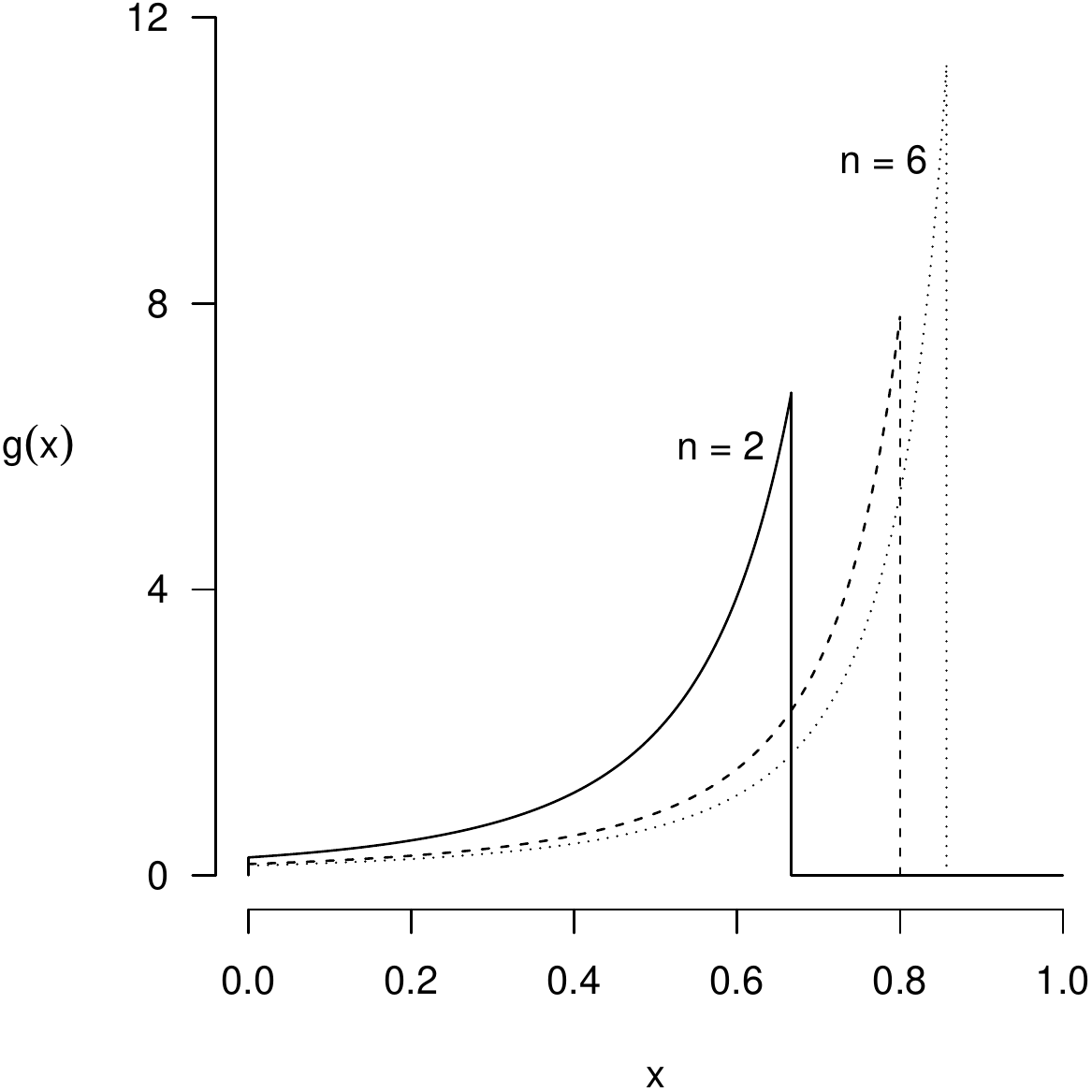}
\caption{Silent duel firing distance densities}
\label{fig:constsum}
\end{figure}
\par
We have thus given a solution to the general constant-sum versions of the silent duel for multiple players. We can also calculate the equilibrium payoff $v$ (and hence the equilibrium distribution) in the case of `no prize for failure'.
\begin{proposition}
For the prize competition case, $c=0$, the equilibrium score distribution $F$ is given by
\begin{equation}
\label{eq:cforz=0}
F(y) = \frac{p}{\sqrt[n-1]{1 - y}} \text{, for }%
0 \leqslant y \leqslant b = 1 - v \text{,}
\end{equation}
where $1/p$ is the unique solution in $(1, \infty)$ of the polynomial
\begin{align*}
\left( \frac{1}{p} \right) ^{n} &= 1 + n \left( \frac{1}{p} \right) \text{.}
\end{align*}%
The payoff $v = p^{n - 1}$.
\end{proposition}
\begin{proof}
Taking $c=0$ in Theorem~\ref{thm:sdsoln}, we easily find that
\begin{equation*}
f(y) = \frac{p}{n-1} \left( \frac{1}{1 - y} \right)^{\frac{n}{n - 1}} \text{, }%
\end{equation*}%
and so
\begin{equation*}
g(x) = \frac{p}{n-1} \left( \frac{1}{1 - x} \right)^{\frac{2n - 1}{n - 1}} \text{.}
\end{equation*}%
\end{proof}
\par
\cite{Sakaguchi1978} gives this equilibrium solution, and \cite{HenigONeill1992} show that it is the unique equilibrium. More recently \cite{PresmanSonin2006} have extended the analysis (making one assumption) to the case of many marksmen who may also have different accuracies.
\par
Figure~\ref{fig:nocomp} plots the firing distance densities for the cases $n = 2$ (solid line), $n = 4$ (dashed line), and $n = 6$ (dotted line).
\begin{figure}[htp]
\centering
\includegraphics[scale=0.7, clip]{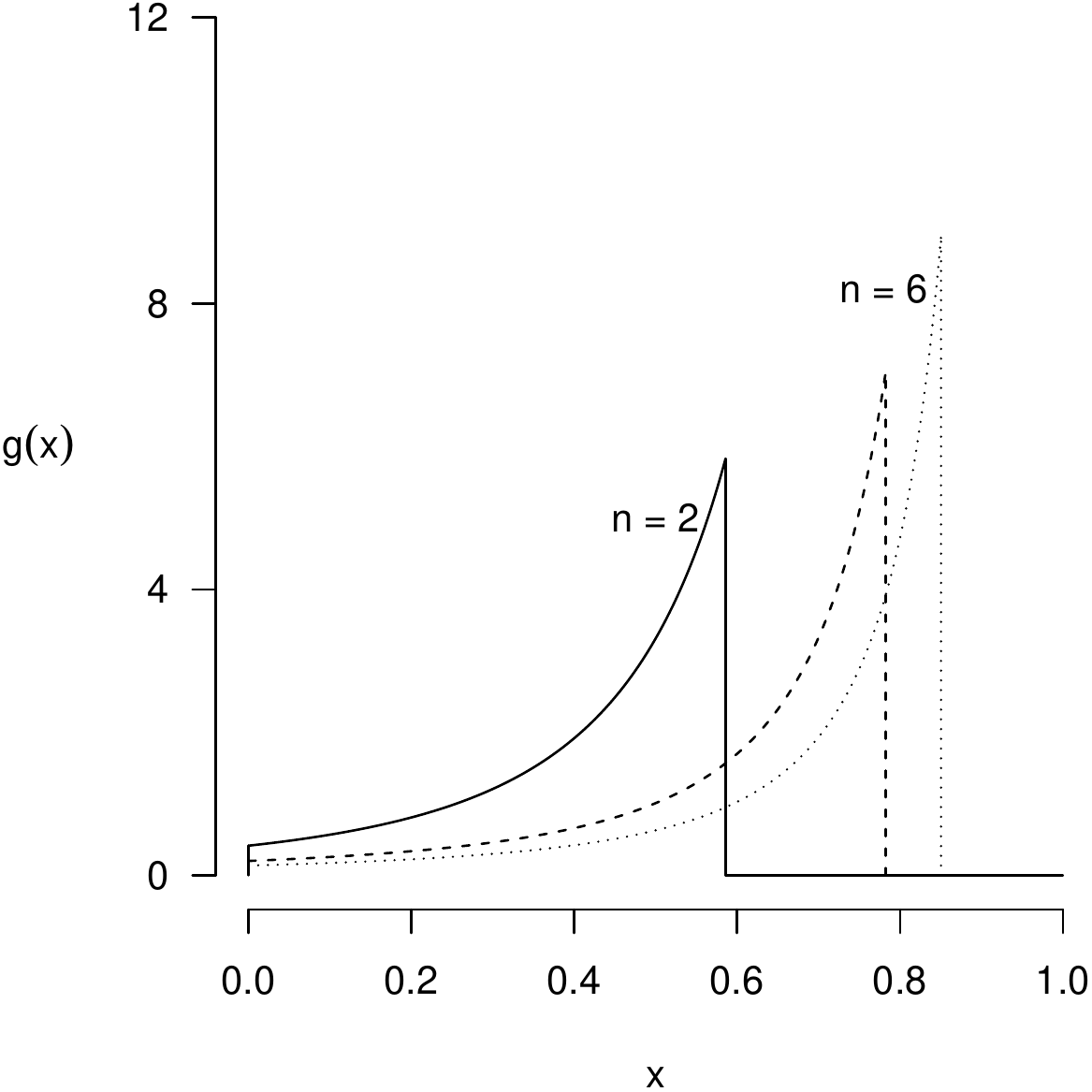}
\caption{Prize competition firing distance densities}
\label{fig:nocomp}
\end{figure}
\section{Conclusion}
\label{sec:conclusion}
We have given a very short solution to the many-player silent duel for with an arbitrary consolation prize. For different values of the consolation prize, our derivation gives or generalizes known results. Ours is the first solution for an arbitrary consolation prize.
%\newpage
%\begin{singlespace}
%
\bibliographystyle{ecta}            %
\bibliography{DRGames}
%
%\end{singlespace}
%

\end{document}